\documentclass[a4paper,12pt]{article}
\usepackage[utf8]{inputenc}
\usepackage{amsmath,amssymb,amsthm,mathtools}
\usepackage{physics}
\usepackage{easyReview}
\usepackage{comment}
\usepackage{hyperref}
\usepackage{todonotes}
\usepackage{authblk}
\usepackage{xspace}
\usepackage{soul}
\usepackage{subfig}

\usepackage{tikz}
\usetikzlibrary{arrows,shapes,positioning,calc}
\tikzset{VertexStyle/.style = {
		shape          = circle,
		draw}}
\tikzset{EdgeStyle/.style   = {}}

\theoremstyle{plain}
\newtheorem{theorem}{Theorem}

\newtheorem{proposition}[theorem]{Proposition}
\newtheorem{lemma}[theorem]{Lemma}
\newtheorem*{definition}{Definition}
\theoremstyle{definition}

\newcommand{\mailto}[1]{\href{mailto:#1}{#1}}

\newcommand{\R}{\mathbb R}

\newcommand{\antigirth}{{\bar{g}}}
\newcommand{\ER}{Erd\H{o}s-R\'enyi\xspace}
\newcommand{\WS}{Watts-Strogatz\xspace}
\newcommand{\BA}{Barab\'asi-Albert\xspace}

\makeatletter
\renewcommand*\env@matrix[1][*\c@MaxMatrixCols c]{%
	\hskip -\arraycolsep
	\let\@ifnextchar\new@ifnextchar
	\array{#1}}
\makeatother

\title{Upperbounds on the probability of finding marked connected components
	using quantum walks}
\author[1,2]{Adam Glos\thanks{\mailto{aglos@iitis.pl}}}
\author[3]{Nikolajs Nahimovs\thanks{corresponding author,
		\mailto{nikolajs.nahimovs@lu.lv}}}
\author[4]{Konstantin Balakirev\thanks{\mailto{kvbalakirev@gmail.com}}}
\author[5,6]{Kamil Khadiev\thanks{\mailto{kamil.hadiev@kpfu.ru}}}

\date{} 
\affil[1]{Institute of Theoretical and Applied Informatics, Polish Academy of 
	Sciences \\Ba{\l}tycka 5, 44-100 Gliwice, Poland}
\affil[2]{Institute of Informatics, Silesian University of Technology\\
	Akademicka 16, 44-100 Gliwice, Poland}
\affil[3]{Center for Quantum Computer Science, Faculty of Computing, University
	of Latvia\\Raina bulv. 19, Riga, LV-1586, Latvia}
\affil[4]{Higher School of Economics\\ Kochnovskiy passage 3, Moscow, Russia}
\affil[5]{Kazan Federal University\\ Kremlevskaya 18, Kazan, Russia}
\affil[6]{Smart Quantum Technologies Ltd.\\ K. Marksa 5, Kazan, Russia}

\begin{document}

\maketitle

	
	\begin{abstract}
	\noindent Quantum walk search may exhibit phenomena beyond the intuition from a
conventional random walk theory. One of such examples is exceptional
configuration phenomenon -- it appears that it may be much harder to find any
of two or more marked vertices, that if only one of them is marked. In this
paper, we analyze the probability of finding any of marked vertices in such
scenarios and prove upper bounds for various sets of marked vertices. We apply
the upper bounds to large collection of graphs and show that the quantum search
may be slow even when taking real-world networks.
	\end{abstract}
	
	
	
	
	
	\section{Introduction}
	
	Quantum walks are quantum counterparts of classical random walks
	\cite{Portugal:2013}. 
	Similarly to classical random walks, there are two types of quantum walks:
	discrete-time quantum walks (DTQW), introduced by Aharonov~{\it et
		al.}~\cite{Aharonov:1993}, and continuous-time quantum walks (CTQW), introduced
	by Farhi~{\it et al.}~\cite{Farhi:1998}.
	For the discrete-time version, the step of the quantum walk is usually given by
	two operators -- coin and shift -- which are applied repeatedly. 
	The coin operator acts on the internal state of the walker and rearranges the
	amplitudes of going to adjacent vertices. The shift operator moves the walker
	between the adjacent vertices.
	
	Quantum walks have been useful for designing algorithms for a variety of search
	problems \cite{Nagaj:2011}.
	To solve a search problem using quantum walks, we introduce the notion of marked
	elements (vertices), corresponding to elements of the search space that we want
	to find.
	We perform a quantum walk on the search space with one transition rule at the
	unmarked vertices, and another transition rule at the marked vertices. If this
	process is set up properly, it leads to a quantum state in which the marked
	vertices have higher probability to be found than the unmarked ones. This method
	of search using quantum walks was first introduced in \cite{Shenvi:2003} and has
	been used many times since then.
	
	In contrary to classical random walks, the behavior of the quantum walk can
drastically change if the search space contains more than one marked element.
In 2008 Ambainis and Rivosh~\cite{Ambainis:2008} have studied DTQW on
two-dimensional grid and showed that if the diagonal of the grid is fully
marked then the probability of finding a marked element does not grow over
time. Later, in 2015 Wong and Ambainis~\cite{Wong:2015} have analysed DTQW on
the simplex of complete graphs and showed that if one of the complete graphs is
fully marked then there is no speed-up over classical exhaustive search. In
both cases the configuration consists of $\Theta(\sqrt{N})$ marked vertices.
The same year Nahimovs and Rivosh~\cite{Nahimovs:2015a,Nahimovs:2015} have
studied DTQW on two-dimensional grid for various placements of multiple marked
vertices and demonstrated configurations of a constant number of marked
vertices (naming them \textit{exceptional configurations}) for which the walk
have no speed-up over classical exhaustive search. Later, Nahimovs and
Santos~\cite{Nahimovs:2017} have extended the results to general graphs.  
Note, that the existence of exceptional configurations is not a universal phenomenon, but a feature of a DTQW model with specific coin operator~\cite{Wong:2016a}.

Recently Ambainis~{\it et al.} \cite{ambainis2019quadratic} proposed a quantum walk search algorithm, that is known to be quadratically faster than regular random walks in finding any set of marked vertices. Despite
this it is still possible to utilize exceptional configurations. The
phenomenon was used to show that quantum algorithms are not secure against
malicious input modification \cite{glos2019impact}, and to detect perfect
matching in a bipartite graph \cite{nahimovs2019adjacent}.
	
	The reason why some configurations are exceptional is that for such
	configurations the initial state of the algorithm is close to a 1-eigenvector of
	a step of the walk algorithm. Therefore, the probability of finding a marked
	vertex stays close to the initial probability and does not grow over time.
	Nahimovs, Khadiev and Santos~\cite{Khadiev:2018} analysed the search for a set
	of connected marked vertices forming an exceptional configuration and proved the
	upper bound on the probability of finding a marked vertex. The proved bound,
	however, depends on a parameter -- a sum of squares of amplitudes of edges
	between the marked vertices inside the stationary state -- which was left
	unestimated.
	
	In this paper we continue the analysis and prove the upper bound in explicit
	form, which depends on properties of the graph and the configuration of marked
	vertices. In particular we show that useful bounds can be derived based only
	on the order and size of the graph and its marked component, and how strongly
	marked component is connected with rest of the graph.
	Additionally, we analyse several examples of sets of marked vertices and show
	tightness of our bounds in the worst case scenario.
	Finally, we apply the bounds to random graphs and networks representing real-world
	dependencies.

	
	
	\section{Preliminaries} \label{sec:preliminaries}
	
	\subsection{Graph theory preliminaries}\label{sec:appendix-preliminaries}
	
	Let $G=(V,E)$ be a simple graph, by which mean that the graph is undirected, unweighted, contains no loops or multiple edges. In this paper we will consider only simple graphs. We say that the
	graph is bipartite if there exists a non-empty $U \subset V$ such that
	$E\subseteq U\times (V\setminus U)$. We call $U$ a bipartite set. We define
	$I(v)$ to be the set of edges incident to $v$, and $N(v)$ to be the set of
	adjacent vertices (neighbourhood of $v$). Furthermore, we define a degree of the vertex $\deg(v) =
	|I(v)| = |N(v)|$. We call a sequence  
	\begin{equation}
	P_{wv} = (v_1=w,e_1,v_2,e_2,\dots, v_{k-1},e_{k-1},v_k=v),
	\end{equation}
	where $e_i=\{v_i,v_{i+1}\}$ a path, if each vertex appears exactly once, allowing
	$v_1 = v_k$. If $v_1 = v_k$ we call the path a cycle. We call $|P|\coloneqq k-1$
	the length of the path $P$. Let $\bar d$ be the diameter in a graph, i.e. the
	maximum distance between any pair of vertices in $G$. We call vertex $v$ a leaf iff $deg(v) = 1$.
	
	We call a graph $H=(V_H,E_H)$ a subgraph of $G$ iff  $V_H \subseteq V$ and
	$E_H\subseteq E$. Furthermore, a subgraph is a spanning subgraph iff $V_H=V$.
	Note that $G$ is a (spanning) subgraph of $G$ as well. We call subgraph
	$H=(V_H,E_H)$ an induced subgraph of $G$ iff for each $v,w\in
	V_H$ we have $\{v,w\}\in E_H \iff \{v , w\}\in E_G$.
	
	We call graph a  forest if it does not contain any cycles. A tree is a connected forest. We say graph is unicyclic if it contains exactly one cycle. Every connected unicyclic graph has precisely $|V|$ edges. One can prove the following
	theorem.
	\begin{lemma}
		Let $G=(V,E)$ be a connected undirected graph with a subgraph $C$ being a
		cycle. Then there exists a connected, unicyclic, spanning subgraph $H\subseteq G$
		such that $C\subset H$.
	\end{lemma}
	\begin{proof}
		The proof is constructive. Let us start with $H_0=G$. If $H_0$ contains
precisely $|V|$ edges, then it matches the requirements of the theorem. If
not, it means it has another cycle $C_1$. Let $e$ be an edge that is in $C_1$
but not in $C$, Then we construct new graph $H_1=(V,E_{H_0}\setminus\{e\})$.
Removing the edge from a cycle does not break connectivity, and still we have
$C\subseteq H_1 \subseteq G$, but with $H_1$ having one edge less. By repeating the procedure for
	$k=2,\dots,|E|-|V|$, we construct connected spanning subgraphs $H_k$ of $G$
	having  $|E|-k$ edges and satisfying $C\subseteq H_k$. By construction
	$H\subseteq H_{|E|-|V|}$ has $|V|$ edges, hence is a unicyclic graph and satisfies requirement from the lemma. 
	\end{proof}
	
	
	\subsection{Quantum walks}
	
	Let $G=(V,E)$ be a simple, undirected graph.
	We define a location register with basis states $\ket{v}$ for $v \in
	V$ and a direction or a coin register, which for a vertex $v$ has
	$\deg(v) = |N(v)|$ basis states $\ket{w}$ for $w \in N(v)$.
	The state of the quantum walk is given by
	\begin{equation}
	\ket{\psi(t)} = \sum_{v\in V} \sum_{w \in N(v)} \alpha_{v,w} \ket{v,w}.
	\end{equation}
	One can consider the state of the quantum walk as a state spanned by a directed graph. Indeed, a single basic state $\ket{v,w}$ has a natural interpretation as arc $(v,w)$ \cite{sadowski2016lively}. Note that in general $\alpha_{v,w}\neq \alpha_{w,v}$.
	
	A step of the quantum walk is performed by first applying $C = I \otimes C'$,
	where $C'$ is a unitary transformation on the coin register. The usual choice of
	transformation on the coin register is Grover's diffusion transformation $D$.
	Then, the shift transformation $S$ is applied,
	\begin{equation}
	S = \sum_{v\in V} \sum_{w \in N(v)} \ket{w,v} \bra{v,w} ,
	\end{equation}
	which for each pair of connected vertices $v,w$ swaps an amplitude of arc $(v,w)$ with an amplitude of arc $(w,v)$. 
	
	The quantum walk starts in the equal superposition over all vertex-direction
	pairs
	\begin{equation}
	\ket{\psi_0} = \frac{1}{\sqrt{2|E|}} 
	\sum_{v\in V} \sum_{w \in N(v)} \ket{v,w}.
	\end{equation}
	It can be verified that the state stays unchanged, regardless of the
	number of steps. 
	
	To use the quantum walk as a tool for search, we introduce the notion of marked
	vertices.
	We perform the quantum walk with one set of transformations at unmarked
	vertices, and another set of transformations at marked vertices. 
	Usually the separation between marked and unmarked vertices is given by the
	query transformation $Q$, which flips the sign at a marked  vertex, irrespective
	of the coin state, i.e. $Q\ket{v,w} = -\ket{v,w}$ iff $v$ is marked. In this
	case a step of the quantum search is given by the transformation $U = SCQ$.
	
	The running time of the walk and the probability of finding a marked vertex in
	general case depends on both the structure of the graph and the configuration,
	i.e. the number and the placement of marked vertices.
	
	
	\subsection{Stationary states}
	
	We call a state of the quantum walk $\ket{\psi_{ss}}$ a stationary state if it
is not changed by the step of the walk $U$, namely $\ket{\psi_{ss}} =  U\ket{\psi_{ss}}$ \footnote{
The existence of stationary states is not a universal phenomenon, but a feature of a quantum walk model. For DTQW model it depends on a coin transformation used by the walk~\cite{Wong:2016a}.}.
The step of the walk depends on a set of marked vertices $V_M$.
In order to identify sets of marked vertices which result in inefficient quantum search, we need to consider stationary
states with overlap $|\braket{\psi_{ss}}{\psi_0}|$ being as large as possible.
	
It turns out that we may limit ourselves to a special class of stationary states (given by Theorem~\ref{theorem:ss-states}) as there exists a state of this form which has maximal overlap with the initial state among all stationary states (for the given set of marked vertices).

	\begin{theorem}[\cite{Nahimovs:2017,Wong:2016a}]\label{theorem:ss-states}
		Let $G=(V,E)$ be undirected graph. Consider a state $\ket{\psi}=\sum_{v\in V}\sum_{w\in N(v)}\alpha_{v,w}\ket{v,w}$ with the following properties:
		\begin{enumerate}
			\item for all $v\in V$ and $w\in N(v)$ $\alpha_{v,w} = \alpha_{w,v}$,
			\item if $v\in V$ is marked, then $\sum_{w\in N(v)} \alpha_{v,w} = 0$,
			\item if $v \in V$ is unmarked, then for any $w,u\in N(v)$ we have $\alpha_{v,w} = \alpha_{v,u}$.
		\end{enumerate}
		Then $\ket{\psi}$ is a stationary state of the quantum search operator $U$.
	\end{theorem}

\noindent
Fig.~\ref{fig:ss_example} shows two different stationary states for the same set of marked vertices for a two-dimensional grid.

\begin{figure}
\hfill
\subfloat[]{
	\includegraphics[scale=0.7]{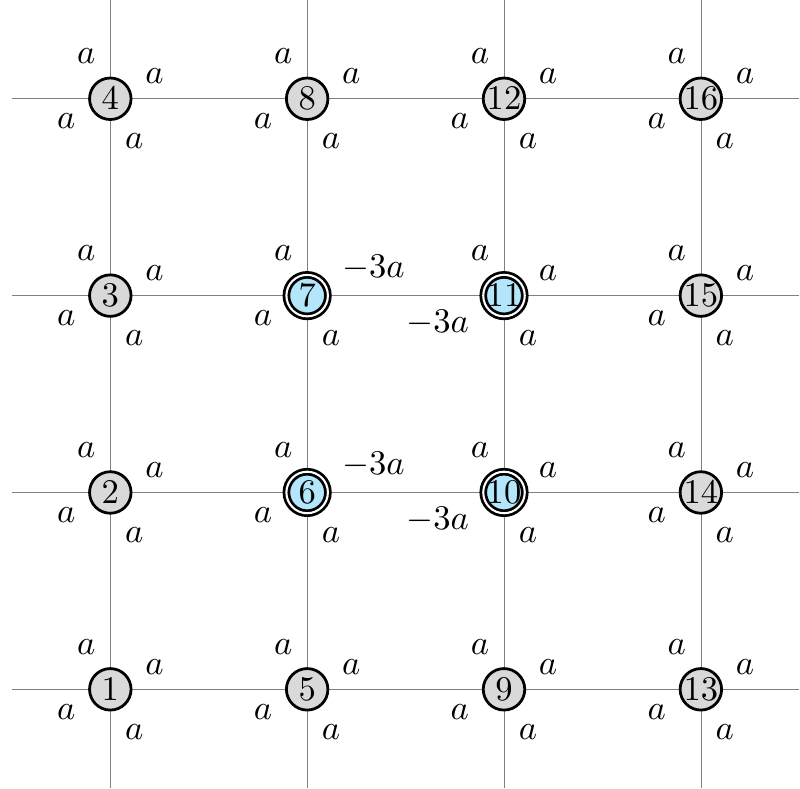}}
\hfill
\subfloat[]{
	\includegraphics[scale=0.7]{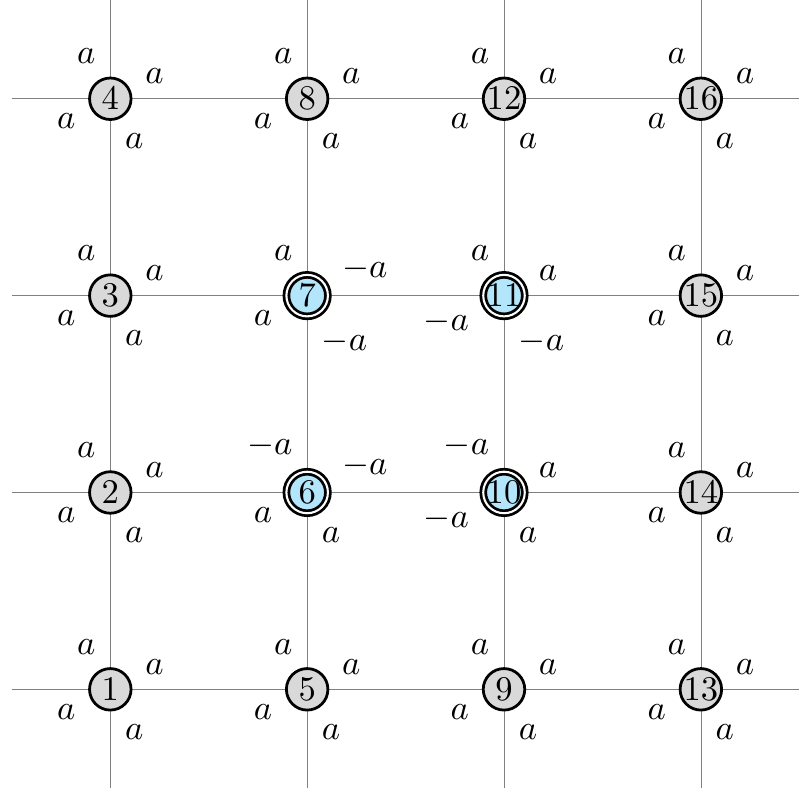}}
\hfill
\caption{\label{fig:ss_example}Stationary states for the two-dimensional grid. Marked vertices are colored in blue.}
\end{figure}

	\begin{theorem}[\cite{Wong:2016a}]\label{theorem:ss-minimum}
		There exists a stationary state of the form as in
		Theorem~\ref{theorem:ss-states} that maximizes the overlap between it and the
		initial state over all possible stationary states.
	\end{theorem}

The following theorem gives the conditions for a configuration (set) of marked vertices to have a stationary state:

\begin{theorem}[\cite{Wong:2016a}]\label{theorem:ss-existence}
	Let $G=(V,E)$ be an undirected graph with a set of marked vertices $V_M$ and let $G_M=(V_M, E_M)$ be its subgraph induced by $V_M$. Then the step of the quantum walk has a stationary state if:
	\begin{enumerate}
		\item $G_M$ is non-bipartite;
		\item $G_M$ is bipartite with bipartite set $U_M$, and 
			\begin{equation}
			\sum_{u\in U_M} \deg_G(u) = \sum_{v\in V_M\setminus U_M} \deg_G(v),
			\end{equation}
		where $\deg_G(v)$ is the degree of vertex $v$ in the graph $G$.			
	\end{enumerate}
\end{theorem}
\noindent Note that this class of configurations of marked vertices does not exhaust all possible configurations with small success probability \cite{nahimovs2019comment}.
To simplify the further discussion we give the following definition.

\begin{definition}
	Let $G=(V,E)$ be an undirected graph, let $V_M$ be a set of marked vertices with $1<|V_M|<|V|$. Let $G_M=(V_M,E_M)$ be a connected induced graph of $G$ that satisfies the items 1 or 2 from  Theorem~\ref{theorem:ss-existence}. In this case we call $G_M$ a Single-Amplitude Marked Component (SAMC) of $G$.
\end{definition}
Consider a SAMC $G_M=(V_M,E_M)$ of graph $G=(V,E)$. If the  subgraph induced by $V\setminus V_M$ is connected, then from Theorem~\ref{theorem:ss-states} (properties 1 and 3) it follows that all amplitudes of a stationary state outside the $G_M$ must be equal. If the subgraph is not connected, we can still find a stationary state for which ``outside'' amplitudes are equal \cite{nahimovs2019comment}, however, the state is not guaranteed to have maximal overlap with the initial state. 

A stationary state of SAMC\footnote{More precisely, of the step of the quantum walk with a set of marked vertices $V_M$.} with large (maximal in case of subgraph induced by $V\setminus V_M$ is connected) overlap  can be written as
\begin{equation}
\ket{\psi_{ss}} = a\sum_{v\in V\setminus V_M}\sum_{w\in N(v)}\ket{v,w} + a\sum_{\{v,w\}\in E_M} c(\{v,w\})
(\ket{v,w}+\ket{w,v})\\
\end{equation}
where $a< 1/\sqrt{2|E|-2|E_M|-D^{\bar M}}$ \cite{nahimovs2019comment}. $D^{\bar M}$ denotes the number of edges from $V_M$ to $V \setminus V_M$. We need to find a mapping $c:E_M\to \R$ (assignment of amplitudes $c(\{v,w\})$) such that the overlap between the initial state and the stationary state above is maximized. Note, that $c$ must satisfy the property 2 of Theorem~\ref{theorem:ss-states}. 

The paper \cite{Khadiev:2018} proves the bound on the success probability of the walk:
	
	\begin{theorem}[\cite{Khadiev:2018,nahimovs2019comment}]\label{theorem:original-upperbound}
		Consider a graph $G=(V,E)$ with a SAMC $G_M = (V_M,E_M)$. Let 
		\begin{equation}
		\ket{\psi_{ss}} = a\sum_{v\in V\setminus V_M}\sum_{w\in N(v)}\ket{v,w} + a\sum_{\{v,w\}\in E_M} c(\{w,v\})
		(\ket{v,w}+\ket{w,v})
		\end{equation}
		be a stationary state of SAMC. Then the probability $p_M(t)$ of finding a marked vertex at time $t$
		satisfies 
		\begin{equation}\label{eq:scucess-probability-general-upperbound}
		\max_{t\geq 0} p_M(t) \leq \frac{4}{2|E|-2|E_M|-D^{\bar M}} \left(\sum_{e\in
			E_M} c^2(e) + 2 D^{\bar M} + 2|E_M|\right).
		\end{equation}
	\end{theorem}
	
\noindent
While $|E|$, $|E_M|$ and $D^{\bar M}$ depend on $G$ and $V_M$ only, there might exist
multiple assignments of $c(e)$ and, therefore, $\sum c^2(e)$ is not uniquely
defined. In the next section, which is the main result of the paper, we will
focus on estimating the sum.


	\section{The upperbounds on success probability for general SAMC}\label{sec:summary}
	In this section we present a results concerning upperbounds for general
	graphs, and examples showing their tightness. Our results are organized as
	follows. In the Sec.~\ref{sec:the-optimization-problem} we formulate notation
	and present the optimization problem which will play a central role in rest of our
	paper. The main results of our paper concerning general graphs are splitted into
	three parts. In Sec.~\ref{sec:trees} we present general results concerning
	trees. In Sec.~\ref{sec:nonbipartite} we present results for nonbipartite
	graphs. In Sec.~\ref{sec:bipartite} we present results concerning bipartite
	graphs. Finally, in Sec.~\ref{sec:overview-general-graphs} we conclude our
	results.
	
	Our proofs are constructive, in the sense that they allow
	constructing the stationary state.
	
	\subsection{The optimization problem}\label{sec:the-optimization-problem}
	In Theorem~\ref{theorem:original-upperbound} we have shown a general upperbound,
	which is true for any mapping $c:E_M\to \R$ satisfying some special requirement.
	Thus our goal is to find such $c$, that the bound is as small as possible. This
	gives us an optimization problem, which we define below.
	
	Let $G_M=(V_M,E_M)$ be a SAMC of some graph $G$. Let 
	\begin{equation}
	\ket{\psi_{ss}} = a\sum_{v\in V\setminus V_M}\sum_{w\in N(v)}\ket{v,w} + a\sum_{\{v,w\}\in E_M} c(\{w,v\})
	(\ket{v,w}+\ket{w,v}).
	\end{equation}
	be a stationary state of the quantum search. Conditions 1.\ and 3.\ from Theorem~\ref{theorem:ss-states} are satisfied by the form of the stationary state. For Condition 2., we have to satisfy for all $v\in V_M$
	\begin{equation}
		\begin{split}
			a |N(v)\setminus V_M| + a \sum_{w\in N(v)\cap V_M} c(\{v,w\}) &=0\\
			\sum_{w\in N(v)\cap V_M} c(\{v,w\}) &= -|N(v)\setminus V_M|.
		\end{split}
	\end{equation}
	
	For simplicity let us define \emph{conditioning function} $K:V_M \to \R$ of the form $K(v) = -|N(v)\setminus V_M|$. We can turn the original problem of minimizing the RHS from the upperbound presented in Theorem~\ref{theorem:original-upperbound} into equivalent, purely classical optimization problem 	
	\begin{equation}\label{eq:optimization-problem}
	\begin{split}
	\operatornamewithlimits{minimize}\limits_{c|c : E_M \to \R} &\quad  \|c\|_2^2\\
	\textrm{subject to} &\quad  \sum_{e\in I(v)\cap E_M} c(e) = -K(v) \quad \textrm{for $v\in V_M$}. \\
	\end{split}
	\end{equation}
	Instead of solving the minimization problem we will construct a $c$, and thus stationary state, with sufficiently small $\|c\|_2^2$. By presenting suitable examples we will show that at least in complexity our derivation is tight in the sense of choice of $c$. 
	
	
	Let us now introduce notation used in the paper. We denote $\|K\|=\max_{v\in
	V_M}|K(v)|$ and $\|K\|_1 = \sum_{v \in V_M} |K(v)|$. Let $V_1\subseteq V_M$,
then we will use notation $\|K\|_1^{V_1} = \sum_{v\in V_1} |K(v)|$. In
particular for disjoint $A,B$ we have $\|K\|^A_1 + \|K\|^B_1 = \|K\|^{A\cup
	B}_1$. In the paper we will use convention $\sum_{i \in \emptyset}f(i) = 0$,
thus $\|K\|_1^\emptyset = 0$ as well. Furthermore $\|K\|^{V_M}_1 = \|K\|_1$.
	
	Note that $G_M$ is SAMC of a connected graph, we have $\|K\| = \Omega(1)$ and
$\|K\|_1 = \Omega(1)$. Furthermore for bipartite SAMC with bipartite set $U_M$
we have based on item 2. from Theorem~\ref{theorem:ss-states}
	\begin{equation}
	\sum_{u\in U_M} K(u) = \sum_{v\in V_M\setminus U_M}	K(v).\label{eq:condition-bipartite-k}
	\end{equation}
	
	
	\subsection{Success probability for trees SAMC}\label{sec:trees}
	Let us start with the preliminaries concerning rooted trees. We call $T=(V,E,v)$
	a rooted tree iff $(V, E)$ is a tree and $v\in V$. For such graph, for each
	vertex $w\in V$ there exists a unique path $P_{wv}$. The length of the path is
	denoted by $d(w,v)$. We call the height of the tree $h(T)$ the length of maximal
	path to its root, i.e.
	\begin{equation}
	h(T) \coloneqq \max_{w\in V} |P_{wv}|.
	\end{equation}
	For each $w\in V\setminus\{v\}$ there exists an edge $e_p(w)$ which is the first
	edge of the path $P_{wv}$, that is the edge which connects the vertex $w$ to its
	parent $p(w)$. Let $N_C(w)$ be a collection of vertices $w$ is a parent of. Let
	$I_C(w)=I(w)\setminus\{e_p(w)\}$ be a collection of edges incident to $w$. Note that for the root we have $N_C(v)=N(v)$ and $I_C(v)=I(v)$.
	
	For any rooted tree we can define a natural partial order $(V,\leqq)$ called
	\emph{tree-order}. We have $w'\leqq w$ iff the path $P_{w'v}$ passes through
	$w$. We denote $D(w)\coloneqq \{u\in V: u\leqq w\}$ a set of all descendants of
	$w$. Note, that we assume $u$ is its own descendant. We call $(V,\leqq^\ast)$ a
	linear extension of the partial order iff the $\leqq^\ast$ is linear and for
	each $w\leqq w'$ we have $w\leqq^\ast w'$. We can create linear extension as follows. First, we list all leaves. Then, we inductively list all neighbors of already listed elements, unless they are already included or unless they are root. At the very end we add the root. Note that the root $v$ is the unique
	maximum for both tree-order and its linear extension. 
	Furthermore, the minimal elements of the tree-order are leafs. 
	Please find an example on Fig.~\ref{fig:rooted-tree}. 
	
	\begin{figure}\centering
		\subfloat[\label{fig:rooted-tree-example}rooted tree]{
			\begin{tikzpicture}
			
			\node[draw,circle,fill=white!40!blue] (1) at (0,0) {1};
			\node[draw,circle] (2) at (1.5,0) {2};
			\node[draw,circle] (3) at (3,0) {3};
			
			\node[draw,circle,fill=white!40!blue] (4) at (1.5,-1) {4};
			\node[draw,circle] (5) at (3,-1) {5};
			
			\node[draw,circle,fill=white!40!blue] (6) at (0,-2) {6};
			\node[draw,circle] (7) at (1.5,-2) {7};
			\node[draw,circle] (8) at (3,-2) {8};
			
			\node[draw,circle] (9) at (4.5,-.5) {9};
			\node[draw,circle,fill=red] (10) at (5.5,-1.3) {10};
			
			\draw[thick] (1) -- (2) -- (3) -- (9) -- (5) -- (4);
			\draw[thick] (6) -- (7) -- (8) -- (10) -- (9);
			
			\node[] (root) at (6.3, -.8) {root};
			\node[] (leaves) at (-1, -1) {leaves};
			\draw[dashed,shorten >= .25cm] (leaves) -- (1);
			\draw[dashed,shorten >= .25cm] (leaves) -- (4);
			\draw[dashed,shorten >= .25cm] (leaves) -- (6);
			
			\draw[thick,->,>=latex] (-1,.8) -- node[above] {the direction of $\leqq$ and
				$\leqq^\ast$ orders} (6, .8) ;
			\end{tikzpicture}}\\
		\subfloat[\label{fig:rooted-tree-example-order}Exemplary $\leqq^\ast$ order]{
			\begin{tikzpicture}

			\node[draw,circle,fill=white!40!blue] (1) at (0,0) {1};
			\node[draw,circle,fill=white!40!blue] (4) at (1.3,0) {4};
			\node[draw,circle,fill=white!40!blue] (6) at (2.6,0) {6};
			
			\node[draw,circle] (2) at (3.9,0) {2};
			\node[draw,circle] (3) at (5.2,0) {3};
			\node[draw,circle] (5) at (6.5,0) {5};
			\node[draw,circle] (7) at (7.8,0) {7};
			\node[draw,circle] (8) at (9.1,0) {8};
			
			\node[draw,circle] (9) at (10.4,0) {9};
			\node[draw,circle,fill=red] (10) at (11.7,0) {10};
			
			\draw[thick,->, >=latex] (1) -- (4);
			\draw[thick,->, >=latex] (4) -- (6);
			\draw[thick,->, >=latex] (6) -- (2);
			\draw[thick,->, >=latex] (2) -- (3);
			\draw[thick,->, >=latex] (3) -- (5);
			\draw[thick,->, >=latex] (5) -- (7);
			\draw[thick,->, >=latex] (7) -- (8);
			\draw[thick,->, >=latex] (8) -- (9);
			\draw[thick,->, >=latex] (9) -- (10);
			\end{tikzpicture}}
		\caption{An example of \protect\subref{fig:rooted-tree-example} rooted tree,
			\protect\subref{fig:rooted-tree-example-order} and its exemplary linear order $\leqq^\ast$ . Leaves
			are marked blue, and root is marked red.}\label{fig:rooted-tree}
	\end{figure}
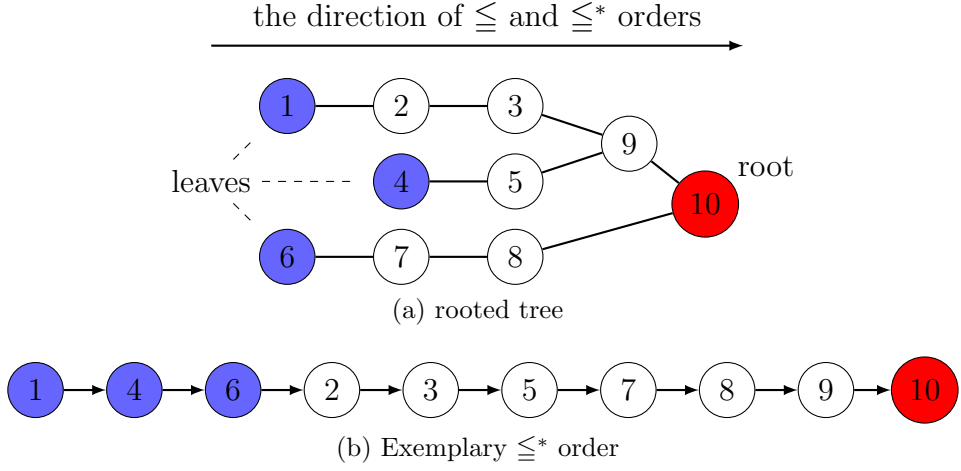
	
	Let us start with simple graph-theoretic bound. While the lemma does not provide
	immediate application in quantum walk theory, we leave the proof here so that
	the reader get used to concept of linear extension of tree order, and the
	induction made on it. 
	\begin{lemma}\label{theorem:tree-sum-descendants-squared}
		Let $T(V,E,v)$ be a rooted tree with height $h$. Then
		\begin{equation}
		\sum_{u\in V\setminus\{v\}} |D(u)|^2 \leq |V|^2h.
		\end{equation}
	\end{lemma}
	\begin{proof}
		Note that for $V=\{v\}$ tree has height 0, thus the statement is true. 	
		We will show the statement by proving inductively 
		\begin{equation}
		\sum_{w\in D(u)\setminus\{u\}} |D(w)|^2 \leq h_u|D(u)|^2,
		\end{equation}
		for all $u\in V$ in linear
		extension of tree order $\leqq^\ast$. Here $h_w$ is the longest path from $w$ to its descendant. Note that for $u=v$
		we have the original statement. The equation above is true for $u$ being a leaf.
		For other nodes we have
		\begin{equation}
		\begin{split}
		\sum_{w\in D(u)\setminus\{u\}} |D(w)|^2 &= \sum_{w\in N_C(u)} \sum_{t\in
			D(w)\setminus\{w\}}|D(t)|^2 + \sum_{w\in N_C(u)} |D(w)|^2 \\
		& \leq \sum_{w\in N_C(u)}h_w|D(w)|^2 + \sum_{w\in N_C(u)} |D(w)|^2 \\
		& \leq (h_u-1)\sum_{w\in N_C(u)}|D(w)|^2 + \sum_{w\in N_C(u)} |D(w)|^2 \\
		& \leq h_u\sum_{w\in N_C(u)}|D(w)|^2 \leq h_u\left (\sum_{w\in
			N_C(u)}|D(w)|\right )^2\\
		& \leq h_u |D(u)|^2.
		\end{split}		
		\end{equation} 
	\end{proof}
	
	We start with a lemma which will be used in most of proofs, including those
	related to graphs not being trees. 
	
	\begin{lemma}\label{theorem:rooted-tree}
		Let $T=(V,E,v)$ be a rooted undirected tree and $K:V\to \R$. Then there is a
		unique function $c:E\to \R$ satisfying condition in Eq.~\eqref{eq:optimization-problem} for all
		vertices except the root. Furthermore it satisfies
		\begin{gather}
		\forall u \in V \quad \Big | \sum_{e\in I_C(u)} c(e) \Big | \leq
		\|K\|^{D(u)\setminus \{u\}}_1,\label{eq:rooted-tree-vertex}\\
		\forall u \in V\setminus\{v\} \quad |c(e_p(u))| \leq \|K\|_1^{D(u)}.
		\label{eq:rooted-tree-edge}
		\end{gather}
	\end{lemma}
	
	\begin{proof}
		For $|V|=1$ Eqs.~\eqref{eq:rooted-tree-vertex}~\eqref{eq:rooted-tree-edge} are
		straightforward. Suppose $|V|>1$.
		
		Let $v_1,\dots,v_n=v$ be vertices enumerated according to some linear extension
		$\leqq^\ast$ of the tree-order of $T$. We will assign the values of
		$c(e_p(v_i))$ and prove Eqs.~\eqref{eq:rooted-tree-vertex} and
		\eqref{eq:rooted-tree-edge} inductively with the given order of vertices. Note
		that uniqueness comes directly from the construction.
		
		Suppose $v_i$ is a leaf. According to condition from Eq.~\eqref{eq:optimization-problem}, we
		have
		\begin{equation}
		\sum_{e \in I(v_i)} c(e)  = c(e_p(v_i)) = -K(v_i).
		\end{equation}
		Note that in this case Eq.~\eqref{eq:rooted-tree-vertex} and
		\eqref{eq:rooted-tree-edge} are trivially fulfilled as $D(v_i)=\{v_i\}$.
		
		Suppose $v_i$ is not a leaf. Then, according to the order $\leqq^\ast$, for all
		edges $e \in I_C(v_i)$ value of $c(e)$ is specified. According to
		condition from Eq.~\eqref{eq:optimization-problem} we have
		\begin{equation}
		\sum_{e \in I(v_i)} c(e)  = \sum_{e\in I_C(v_i)} c(e) + c(e_p(v_i)) = -K(v_i).
		\end{equation}
		From this we have
		\begin{equation} \label{eq:cepvi-simplification}
		c(e_p(v_i)) = -\sum_{e\in I_C(v_i)} c(e)-K(v_i).
		\end{equation}
		Let us start with proving Eq.~\eqref{eq:rooted-tree-vertex} for $v_i$:
		\begin{equation}
		\begin{split}
		\Big | \sum_{e\in I_C(v_i)} c(e)\Big |&\leq  \sum_{e\in I_C(v_i)} |c(e)| =
		\sum_{u\in N_C(v_i)} |c(e_p(u))|\leq \sum_{u\in N_C(v_i)} \|K\|_1^{D(u)} \\
		& = \|K\|_1^{D(v_i)\setminus \{v_i\}},
		\end{split}
		\end{equation}
		where the  second inequality comes from the induction assumption on
		Eq.~\eqref{eq:rooted-tree-edge}, and the last equality results from the fact
		that we enumerated all descendants of $v_i$ except $v_i$. Suppose $v_i$ is not
		the root. We have
		\begin{equation}
		\begin{split}
		|c(e_p(v_i))| &= \Big | -\sum_{e\in I_C(v_i)} c(e)-K(v_i)\Big | \leq \Big | 
		\sum_{e\in I_C(v_i)} c(e) \Big| + |K(v_i)|\\ & \leq  \|K\|_1^{D(v_i)\setminus
			v_i} + |K(v_i)| = \|K\|_1^{D(v_i)}.
		\end{split}
		\end{equation}
	\end{proof}
	
	We would like to emphasize that the condition given in
	Eq.~\eqref{eq:optimization-problem} is not satisfied for the root. However, since any
	tree is a bipartite graph, the $K$ function should satisfy
	Eq.~\eqref{eq:condition-bipartite-k}, which has not been assumed here. 
	Let us present the explicit upper-bound on the $\|c\|^2$ for general tree
	graph.
	\begin{proposition}\label{theorem:tree-not-useful-bound}
		Let $G=(V,E)$ be a tree and a SAMC of some connected graph, and let $K$ be an appropriate  conditioning function. Then there exists a unique solution
		$c$ of optimization problem given in Eq.~\eqref{eq:optimization-problem}. Furthermore,
		for arbitrary rooted tree $T=(V,E,v)$ we have
		\begin{equation}\label{eq:tree-not-useful-bound}
		\|c\|_2^2 \leq \sum_{u\in V\setminus \{v\}}\left (\|K\|_1^{D(u)} \right )^2
		\end{equation}
	\end{proposition}
	\begin{proof}
		Let $v\in V $, and let $T=(V,E,v)$ be a rooted tree with height $h$. Based on
		Lemma~\ref{theorem:rooted-tree} there is unique $c$ satisfying
		condition of Eq.~\eqref{eq:optimization-problem} for every $u\in V\setminus \{v\}$. Hence, we need to
		show that the constraint holds for $v$ as well.
		
		Let $V_k$ be a set of vertices located at distance $k$ from root. Note that 
		\begin{equation}
		\begin{split}
		0&=K(v)+\sum_{e\in I(v)} c(e) =K(v)+ \sum_{w_1\in N_C(v)} c(e_p(w_1)) \\
		&= K(v)+\sum_{w_1\in N_C(v)} \Big (-\sum_{w_2 \in N_C(w_1)} c(e_p(w_2)) -
		K(w_1) \Big )\\
		&= K(v) - \sum_{w_1\in V_1}K(w_1) - \sum_{w_2\in V_2}c(e_p(w_2)) = \dots \\
		&= K(v) - \sum_{w_1\in V_1}K(w_1) + \sum_{w_2\in V_2}K(w_2) \dots
		+(-1)^{h}\sum_{w_h \in V_h}K(w_h),
		\end{split} 
		\end{equation} 
		where second line comes from Eq.~\eqref{eq:cepvi-simplification}.
		We have obtained an equivalent form of condition on $K$ given by
		Eq.~\eqref{eq:condition-bipartite-k}. Thus, $c$ is a proper solution.
		The Eq.~\eqref{eq:tree-not-useful-bound} comes directly from applying
		Lemma~\ref{theorem:rooted-tree}.
	\end{proof}
	
	Finally let us formulate a final result concerning the success probability of
	trees.
	\begin{theorem}\label{theorem:tree-useful-bound}
		Let $G_M=(V_M,E_M)$ be a tree SAMC of connected graph $G=(V,E)$, and let
$K$ be an appropriate conditioning function. Then there exists a unique
solution $c$ of the optimization problem Eq.~\eqref{eq:optimization-problem}.
Then the probability
$p_M(t)$ of finding any element from $V_M$ at time $t$ satisfies
		\begin{equation}
		\max_{t\geq 0} p_M(t) =\order{\frac{\bar d |V_M|^2 \|K\|^2}{2|E|-2|E_M| -
				D^{\bar M}}},
		\end{equation} 
		where $\bar d$ is the diameter of the graph.
	\end{theorem}
	\begin{proof}
		Based on the
		Theorem~\ref{theorem:original-upperbound} we have only to upperbound the part
		$\|c\|_2^2 + 2D^{\bar M}$. 
		Note that $D^{\bar M} \leq |V_M| \|K\|$ based on the definition of $D^{\bar
			M}$. Using Proposition~\ref{theorem:tree-not-useful-bound} we have
		\begin{equation}
		\begin{split}
		\|c\|_2^2& \leq \sum_{u\in V\setminus \{v\}}\left (\|K\|_1^{D(u)} \right )^2
		\leq \sum_{u\in V\setminus \{v\}}\left (\|K\| |D(u)| \right )^2 \\
		& \leq \|K\|^2 \sum_{u\in V\setminus \{v\}}|D(u)|^2 \leq \|K\|^2|V_M|^2  \bar
		d, 
		\end{split}
		\end{equation}
		where last inequality comes from
		Lemma~\ref{theorem:tree-sum-descendants-squared} and the fact that for any
		rooted tree the height is smaller than its diameter. Note that $D^{\bar M} =
		\order{\|c\|_2^2}$, and $|E_M|=|V_M|-1$, which ends the proof.
	\end{proof}
	Note that RHS will converge to zero if $|E|$ will be sufficiently large. This is
	in fact general result that will happen in further graphs as well.
	
	The upperbound presented in the theorem above may be inexact for two reasons.
	First, the original bound from Theorem~\ref{theorem:original-upperbound} may be
	not tight in general. Second, the bound on $\|c\|_2^2$ may be not sufficiently
	tight. However, we will show that in the worst case scenario the latter bound is
	tight.
	
	\begin{figure}\centering
		\begin{tikzpicture}
		\node[VertexStyle] (t1) at (0,2) {$t_1$};
		\node[VertexStyle] (t2) at (0,1) {$t_2$};
		\node[VertexStyle] (t3) at (0,0) {$t_3$};
		\node[] at (0,-.75) {$\vdots$};
		\node[VertexStyle] (tnp) at (0,-1.75) {$t_{n'}$};
		
		\node[VertexStyle] (u1) at (7.5,2) {$u_1$};
		\node[VertexStyle] (u2) at (7.5,1) {$u_2$};
		\node[VertexStyle] (u3) at (7.5,0) {$u_3$};
		\node[] at (7.5,-.75) {$\vdots$};
		\node[VertexStyle] (unp) at (7.5,-1.75) {$u_{n'}$};
		
		\node[VertexStyle] (v1) at (1.5, 0) {$v_1$};
		\node[VertexStyle] (v2) at (3, 0) {$v_2$};
		\node[] (vdots) at (4.5, 0) {$\cdots$};
		\node[VertexStyle] (vq) at (6, 0) {$v_{\bar d-1}$};
		
		\draw[EdgeStyle] (v1) to (v2) {};
		\draw[EdgeStyle] (v2) to (vdots) {};
		\draw[EdgeStyle] (vq) to (vdots) {};
		
		\draw[EdgeStyle] (v1) to (t1) {};
		\draw[EdgeStyle] (v1) to (t2) {};
		\draw[EdgeStyle] (v1) to (t3) {};
		\draw[EdgeStyle] (v1) to (tnp) {};
		
		\draw[EdgeStyle] (vq) to (u1) {};
		\draw[EdgeStyle] (vq) to (u2) {};
		\draw[EdgeStyle] (vq) to (u3) {};
		\draw[EdgeStyle] (vq) to (unp) {};
		\end{tikzpicture}
		
		\caption{\label{fig:glued-stars}Glued stars with maximal $2n'+\bar d-1$ number
			of nodes and $d'$ diameter}
		
	\end{figure}
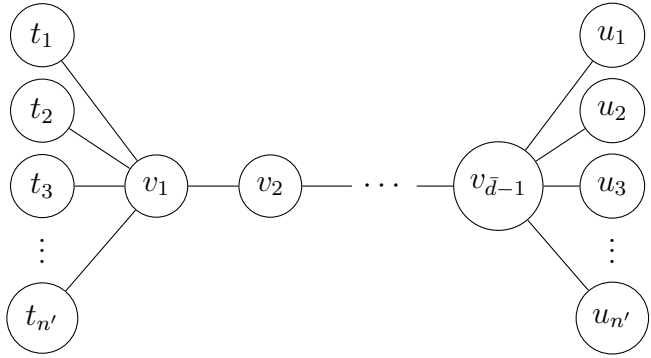
	
	Let $\bar d$ be odd and $n>\bar d$ be even. We consider the following graph with $n$ vertices.
	Let $v_1,\dots,v_{\bar d-1}$ be vertices forming a path in the given order, and
	let $t_1,\dots,t_{n'}$ and $u_1,\dots,u_{n'}$ for $n'=\frac{n-\bar d+1}{2}$ be
	vertices connected to $v_1$ and $v_{\bar d - 1}$ respectively, as in
	Fig.~\ref{fig:glued-stars}. 
	Let $K:v\mapsto k>0$ be a constant conditioning function. Since the graph is a tree, by the
	Theorem~\ref{theorem:tree-not-useful-bound} the function $c$ is unique and thus
	optimal.
	
	The solution function $c$ takes the form
	\begin{equation}
	c(e) = \begin{cases}
	-k, & e=\{v_1,t_i\},\\
	-k, & e=\{v_{\bar d-1},u_i\},\\
	(-1)^{i+1}(n'-1)k, & e=\{v_i,v_{i+1}\}.
	\end{cases}
	\end{equation}
	Thus, we have
	\begin{equation}
	\begin{split}
	\|c\|_2^2 &= 2n'k^2 + (\bar d-2)((-1)^{i+1}(n'-1)k)^2\\
	& \geq (\bar d-2)(n'-1)^2k^2 
	\end{split}	
	\end{equation}
	Note that $n' = \frac{n-\bar d + 1}{2}$, hence $n' = \Theta(n-\bar d)$. If we
	choose $\bar d\leq n/2$, then the graph satisfies the theorem statement.

	
	
	\subsection{Success probability for nonbipartite SAMC}\label{sec:nonbipartite}
	\begin{figure}[t]
		\begin{tikzpicture}[yscale=.9]
		\node[draw,fill=black,circle] (1a) at (-1,0) {};
		\node[draw,fill=black,circle] (2a) at (-.5,1) {};
		\node[draw,fill=black,circle] (3a) at (0.5,1) {};
		\node[draw,fill=black,circle] (4a) at (1,0) {};
		\node[draw,fill=black,circle] (5a) at (.5,-1) {};
		\node[draw,fill=black,circle] (6a) at (-.5,-1) {};
		\draw[ultra thick] (1a) -- (2a) -- (3a) -- (4a) -- (5a) -- (6a) -- (1a);
		\draw[ultra thick] (1a) -- (3a) -- (6a);
		\draw[ultra thick] (3a) -- (5a) -- (2a);
		
		\node[draw,fill=black,circle] (1aa) at (2,.5) {}; 
		\node[draw,fill=black,circle] (2aa) at (1.5,1.5) {};
		\node[draw,fill=black,circle] (3aa) at (-1.5,1.5) {};
		\draw[ultra thick] (2aa) -- (3a) -- (1aa) -- (4a) -- (1aa) -- (2aa);
		\draw[ultra thick] (2aa) -- (3aa) -- (1a);
		\draw[ultra thick] (2a) -- (3aa);
		
		\node[draw,fill=black,circle] (4aa) at (0,-2) {};
		\node[draw,fill=black,circle] (5aa) at (-1,-2.5) {};
		\draw[ultra thick] (1a) -- (5aa) -- (6a) -- (4aa) -- (1aa);
		\draw[ultra thick] (5a) -- (4aa) -- (5aa) -- (3aa);

		\draw[ultra thick,->,>=latex] (3,0) --node[above] {find} node[below] {unicyclic
			graph} (7,0);
		\coordinate (shiftb) at (9,0);
		\node[draw,fill=black,circle] (1b) at ($(1a) + (shiftb)$) {};
		\node[draw,fill=black,circle] (2b) at ($(2a) + (shiftb)$) {};
		\node[draw,fill=black,circle] (3b) at ($(3a) + (shiftb)$) {};
		\node[draw,fill=black,circle] (4b) at ($(4a) + (shiftb)$) {};
		\node[draw,fill=black,circle] (5b) at ($(5a) + (shiftb)$) {};
		\node[draw,fill=black,circle] (6b) at ($(6a) + (shiftb)$) {};
		\draw[ultra thick] (1b) -- (2b) -- (3b) -- (4b) -- (5b) -- (6b) -- (1b);
		\draw[dashed] (1b) -- (3b) -- (6b);
		\draw[dashed] (3b) -- (5b) -- (2b);
		
		\node[draw,fill=black,circle] (1bb) at ($(1aa) + (shiftb)$) {}; 
		\node[draw,fill=black,circle] (2bb) at ($(2aa) + (shiftb)$) {};
		\node[draw,fill=black,circle] (3bb) at ($(3aa) + (shiftb)$) {};
		\draw[ultra thick] (1bb) -- (2bb) -- (3b);
		\draw[ultra thick] (2bb) -- (3bb);
		\draw[dashed] (3b) -- (1bb) -- (4b) -- (1bb);
		\draw[dashed] (3bb) -- (1b);
		\draw[dashed] (2b) -- (3bb);
		
		\node[draw,fill=black,circle] (4bb) at ($(4aa) + (shiftb)$) {};
		\node[draw,fill=black,circle] (5bb) at ($(5aa) + (shiftb)$) {};
		\draw[ultra thick] (6b) -- (4bb) -- (5bb);
		\draw[dashed] (1b) -- (5bb);
		\draw[dashed] (4bb) -- (1bb);
		\draw[dashed] (5b) -- (4bb) -- (5bb) -- (3bb);
		
		\coordinate (shiftc) at (0,-6);
		\draw[ultra thick,->,>=latex, shorten >= 3.5cm, shorten <= 2.7cm] (shiftb)
		--node[sloped,above] {set values} node[sloped,below] {outside the circle}
		(shiftc);
		
		\node[draw,fill=black,circle] (1c) at ($(1a) + (shiftc)$) {};
		\node[draw,fill=black,circle] (2c) at ($(2a) + (shiftc)$) {};
		\node[draw,fill=black,circle] (3c) at ($(3a) + (shiftc)$) {};
		\node[draw,fill=black,circle] (4c) at ($(4a) + (shiftc)$) {};
		\node[draw,fill=black,circle] (5c) at ($(5a) + (shiftc)$) {};
		\node[draw,fill=black,circle] (6c) at ($(6a) + (shiftc)$) {};
		\draw[ultra thick] (1c) -- (2c) -- (3c) -- (4c) -- (5c) -- (6c) -- (1c);
		\draw[dashed] (1c) -- (3c) -- (6c);
		\draw[dashed] (3c) -- (5c) -- (2c);
		
		\node[draw,fill=black,circle] (1cc) at ($(1aa) + (shiftc)$) {}; 
		\node[draw,fill=black,circle] (2cc) at ($(2aa) + (shiftc)$) {};
		\node[draw,fill=black,circle] (3cc) at ($(3aa) + (shiftc)$) {};
		\draw[ultra thick,blue] (1cc) -- (2cc) -- (3c);
		\draw[ultra thick,blue] (2cc) -- (3cc);
		\draw[dashed] (3c) -- (1cc) -- (4c) -- (1cc);
		\draw[dashed] (3cc) -- (1c);
		\draw[dashed] (2c) -- (3cc);
		
		\node[draw,fill=black,circle] (4cc) at ($(4aa) + (shiftc)$) {};
		\node[draw,fill=black,circle] (5cc) at ($(5aa) + (shiftc)$) {};
		\draw[ultra thick,blue] (6c) -- (4cc) -- (5cc);
		\draw[dashed] (1c) -- (5cc);
		\draw[dashed] (4cc) -- (1cc);
		\draw[dashed] (5c) -- (4cc) -- (5cc) -- (3cc);
		
		\coordinate (shiftd) at (9,-6);
		\draw[ultra thick,->,>=latex] (3,-6) --node[sloped,above] {set values}
		node[sloped,below] {on the circle} (7,-6);
		
		\node[draw,fill=black,circle] (1d) at ($(1a) + (shiftd)$) {};
		\node[draw,fill=black,circle] (2d) at ($(2a) + (shiftd)$) {};
		\node[draw,fill=black,circle] (3d) at ($(3a) + (shiftd)$) {};
		\node[draw,fill=black,circle] (4d) at ($(4a) + (shiftd)$) {};
		\node[draw,fill=black,circle] (5d) at ($(5a) + (shiftd)$) {};
		\node[draw,fill=black,circle] (6d) at ($(6a) + (shiftd)$) {};
		\draw[ultra thick,blue] (1d) -- (2d) -- (3d) -- (4d) -- (5d) -- (6d) -- (1d);
		\draw[dashed] (1d) -- (3d) -- (6d);
		\draw[dashed] (3d) -- (5d) -- (2d);
		
		\node[draw,fill=black,circle] (1dd) at ($(1aa) + (shiftd)$) {}; 
		\node[draw,fill=black,circle] (2dd) at ($(2aa) + (shiftd)$) {};
		\node[draw,fill=black,circle] (3dd) at ($(3aa) + (shiftd)$) {};
		\draw[ultra thick,blue] (1dd) -- (2dd) -- (3d);
		\draw[ultra thick,blue] (2dd) -- (3dd);
		\draw[dashed] (3d) -- (1dd) -- (4d) -- (1dd);
		\draw[dashed] (3dd) -- (1d);
		\draw[dashed] (2d) -- (3dd);
		
		\node[draw,fill=black,circle] (4dd) at ($(4aa) + (shiftd)$) {};
		\node[draw,fill=black,circle] (5dd) at ($(5aa) + (shiftd)$) {};
		\draw[ultra thick,blue] (6d) -- (4dd) -- (5dd);
		\draw[dashed] (1d) -- (5dd);
		\draw[dashed] (4dd) -- (1dd);
		\draw[dashed] (5d) -- (4dd) -- (5dd) -- (3dd);
		\end{tikzpicture}
		\caption{The sketch of the proof for non-bipartite and bipartite graphs. First,
			unicyclic graph is chosen, all edges not belonging to the unicyclic graphs
			are set to 0. Then, we set all values for edges outside the cycle (marked blue).
			Finally, values for edges on the cycle graphs are set.}
		\label{fig:sketch-of-the-proof}
	\end{figure}
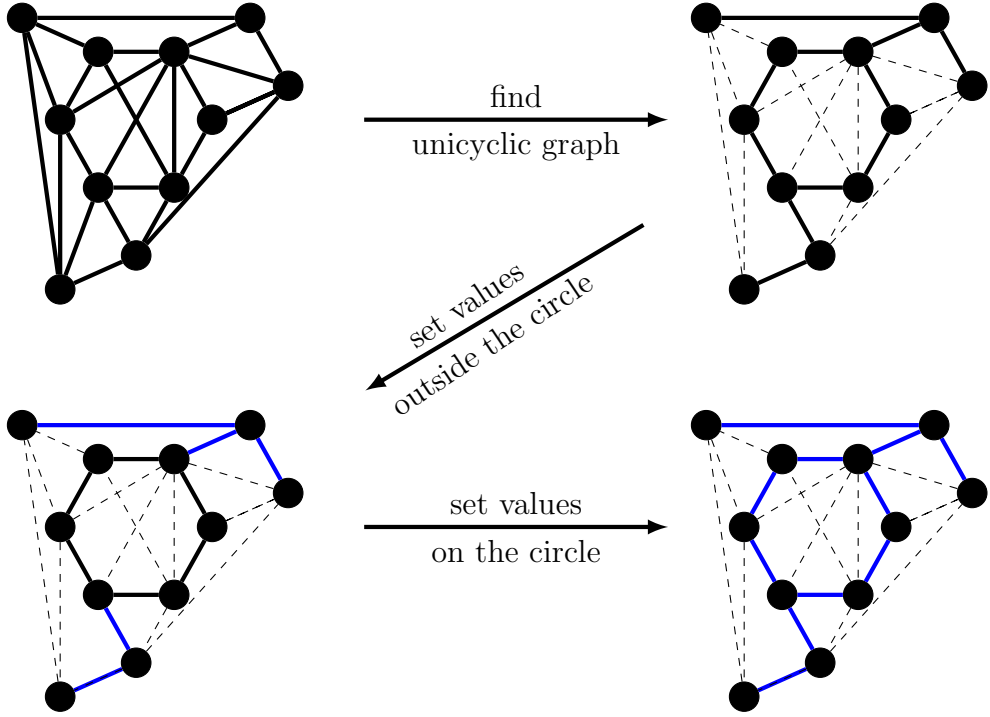

	In this section we will show an upperbound for general non-bipartite SAMC. The
	proof, visualized on Fig.~\ref{fig:sketch-of-the-proof} goes as follows. First,
	we choose an unicyclic spanning subgraph of SAMC. Since SAMC is non-bipartite, we can
	find such graph with cycle of odd length. Then we set the values of edges for
	edges outside the cycle, as it was done in Theorem~\ref{theorem:rooted-tree}.
	Then, we solve the system of linear equations in order to determine the
	values of $c$ for the rest of the edges. Finally, we upperbound the norm of
	$\|c\|_2^2$.
	
	\begin{proposition} \label{theorem:bound-generalk-nonbipartite}
		Let $G=(V,E)$ be a non-bipartite SAMC of some connected graph with an appropriate
conditioning function $K$ and cycle of odd length $\antigirth$. Then there
exists solution $c$ of the optimization problem
Eq.~\eqref{eq:optimization-problem} such that
		\begin{equation}
		\|c\|_2^2 \leq \frac{\antigirth}{4}|V|^2\|K\|^2 +
		(|V|-\antigirth)(|V|-\antigirth+1)^2 \|K\|^2,
		\end{equation}
		and
		\begin{equation}
		\|c\|_2^2 \leq \left (|V|- \frac{3}{4}\antigirth \right )\|K\|_1^2.
		\end{equation}
	\end{proposition}
	\begin{proof}
		Let $C=(V_C,E_C)$ be a cycle of $G$  of length $\antigirth$, with vertices
		$v_1,\dots,v_\antigirth$ with respectively $e_1,\dots,e_{\antigirth}$ where
		$e_j=\{v_{j},v_{j+1}\}$. Let $H=(V,E_H)$ be a connected, unicyclic spanning
		subgraph of $G$ containing $C$. We set $c(e)=0$ for all $e\in E\setminus E_H$.
		
		Now let us consider graph $(V,E_H\setminus E_C)$. Such constructed graph is a
		forest which consists of $\antigirth$ trees such that each tree contains
		precisely one vertex from cycle $C$. Each tree can be considered as an ordered
		tree $T_j=(V_j,E_j,v_j)$ with root $v_j\in V_C$, and by
		Lemma~\ref{theorem:rooted-tree} we can define uniquely $c$ for each tree such
		that
		\begin{gather}
		\forall u \in V_j\setminus\{v_j\} \quad |c(e_p(u))|\leq \|K\|_1^{V_j},
		\label{eq:rooted-tree-edge-nonbip-origin}\\
		\forall u \in V_j \quad \Big | \sum_{e\in I^j_C(u)} c(e)\Big | \leq
		\|K\|_1^{V_j\setminus \{v_j\}}.\label{eq:rooted-tree-vertex-nonbip-origin}
		\end{gather}
		
		Now let us consider the missing part of the domain of $c$, which is $E_C$.
		Based on the condition of Eq.~\eqref{eq:optimization-problem} the values $c(e_j)$ must satisfy
		\begin{equation}
		c(e_j) + c(e_{j+1}) = - \sum_{e\in I_C^{j+1}(v_{j+1})}c(e) -K(v_{j+1})
		\eqqcolon z_{j+1}.
		\end{equation} 
		Since there is odd number of variables, there exists unique solution of the
		form
		\begin{equation}
		c(e_j) = \frac{(-1)^{j}}{2} \left(
		\sum_{k=1}^{j}(-1)^{k}z_k-\sum_{k=j+1}^{\antigirth}(-1)^kz_k\right).
		\end{equation}
		Note that since we have 
		\begin{equation}
		|z_{j}| \leq \Big |\sum_{e\in I_C^j(v_{j})}c(e) \Big| +|K(v_{j})| \leq
		\|K\|_1^{V_j\setminus \{v_j\}} + |K(v_j)| =  \|K\|_1^{V_j},
		\end{equation}
		we have as well
		\begin{equation}
		|c(e_j)| \leq \frac{1}{2}\sum_{k=1}^\antigirth |z_k| \leq \frac12
		\sum_{k=1}^\antigirth  \|K\|_1^{V_k}  = \frac{1}{2}\|K\|_1.
		\end{equation}
		
		Finally using all of the equations above we have
		\begin{equation}
		\begin{split}
		\|c\|_2^2 &= \sum_{j=1}^\antigirth |c(e_j)|^2 + \sum_{j=1}^\antigirth
		\sum_{e\in E_j} |c(e)|^2 \leq \frac{\antigirth}{4}\|K\|_1^2 +
		\sum_{j=1}^\antigirth \sum_{e\in E_j}\left (\|K\|_1^{V_j}\right )^2  \\
		&\leq \frac{\antigirth}{4}\|K\|_1^2 + (|V|-\antigirth) \|K\|_1^2 = \left (|V|-
		\frac{3}{4}\antigirth \right )\|K\|_1^2,
		\end{split}
		\end{equation} 
		where we used  $|V_j|\leq |V|-\antigirth+1$, $\sum_{j=1}^\antigirth |E_j| =
		(|V|-\antigirth)$, and equations given before. Furthermore using $|V_j| \leq
		(|V|-\antigirth+1)$ and $\|K\|_1^A \leq |A| \|K\|$ we have 
		\begin{equation}
		\begin{split}
		\|c\|_2^2 &\leq \frac{\antigirth}{4}\|K\|_1^2 + \sum_{j=1}^\antigirth
		\sum_{e\in E_j}\left (\|K\|_1^{V_j}\right )^2 \\&\leq
		\frac{\antigirth}{4}|V|^2\|K\|^2 + \sum_{j=1}^\antigirth \sum_{e\in
			E_j}(|V|-\antigirth+1)^2\|K\|^2  \\
		&\leq \frac{\antigirth}{4}|V|^2\|K\|^2 + (|V|-\antigirth)
		(|V|-\antigirth+1)^2\|K\|^2.
		\end{split}
		\end{equation} 
	\end{proof}
	
	Contrary to the case of trees, here uniqueness of $c$ function is guaranteed
	only if the marked component is unicyclic itself. Note also that for different
	unicyclic subgraphs we obtain different solutions $c$. Let us now formulate the
	success probability upper-bound for general non-bipartite graph and general $K$
	function.
	
	\begin{theorem}
		Let $G_M=(V_M,E_M)$ be a non-bipartite SAMC of connected graph $G=(V,E)$ with an appropriate
conditioning function $K$. Then the probability
$p_M(t)$ of finding any element from $V_M$ at time $t$ satisfies
		\begin{equation}
		\max_{t\geq 0} p_M(t) =\order{\frac{|V_M|^3\|K\|^2}{2|E|-2|E_M| - D^{\bar M}}}.
		\end{equation} 
	\end{theorem}
	\begin{proof}
		Based on the Theorem~\ref{theorem:original-upperbound} we have to upperbound
the part $\|c\|_2^2 + 2D^{\bar M}$. Note that $D^{\bar M} \leq |V_M| \|K\|$
based on the definition of $D^{\bar M}$. Using
Proposition~\ref{theorem:bound-generalk-nonbipartite} and noting $3 \leq
\antigirth \leq |V_M|$ we have $\|c\|_2^2 = \order{|V_M|^3\|K\|^2}$. Since
$D^{\bar M} = \order{|V_M|\|K\|}$ and $|E_M| = \order{|V_M|^2}$, we obtained the RHS of the upperbound.
\end{proof}

	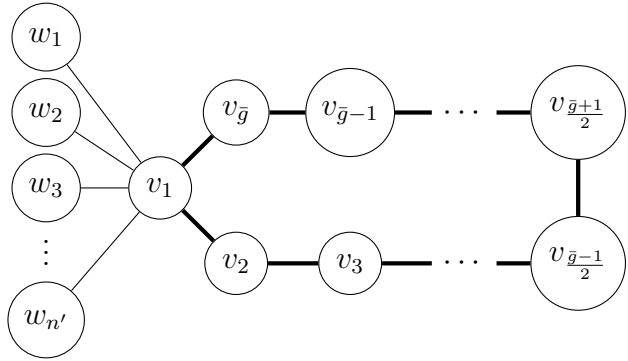
\begin{figure}
		\centering
		\begin{tikzpicture}
		\node[VertexStyle] (w1) at (0,2) {$w_1$};
		\node[VertexStyle] (w2) at (0,1) {$w_2$};
		\node[VertexStyle] (w3) at (0,0) {$w_3$};
		\node[] at (0,-.75) {$\vdots$};
		\node[VertexStyle] (wnp) at (0,-1.75) {$w_{n'}$};
		
		\node[VertexStyle] (v1) at (1.5, 0) {$v_1$};
		\node[VertexStyle] (v2) at (2.5,-1) {$v_2$};
		\node[VertexStyle] (v3) at (4,-1) {$v_3$};
		\node[] (vdotsdown) at (5.5, -1) {$\cdots$};
		\node[VertexStyle] (vqm12) at (7,-1) {$v_{\frac{\antigirth-1}{2}}$};
		\node[VertexStyle] (vqp12) at (7, 1) {$v_{\frac{\antigirth+1}{2}}$};
		\node[] (vdotsup) at (5.5, 1) {$\cdots$};
		\node[VertexStyle] (vqm1) at (4,1) {$v_{\antigirth-1}$};
		\node[VertexStyle] (vq) at (2.5,1) {$v_{\antigirth}$};
		
		\draw[EdgeStyle, ultra thick] (v1) to (v2) {};
		\draw[EdgeStyle, ultra thick] (v2) to (v3) {};
		\draw[EdgeStyle, ultra thick] (v3) to (vdotsdown) {} ;
		\draw[EdgeStyle, ultra thick] (vdotsdown) to (vqm12) {};
		\draw[EdgeStyle, ultra thick] (vqm12) to (vqp12) {};
		\draw[EdgeStyle, ultra thick] (vqp12) to (vdotsup) {};
		\draw[EdgeStyle, ultra thick] (vdotsup) to (vqm1) {};
		\draw[EdgeStyle, ultra thick] (vqm1) to (vq) {};
		\draw[EdgeStyle, ultra thick] (vq) to (v1) {};
		
		\draw[EdgeStyle] (v1) to (w1) {};
		\draw[EdgeStyle] (v1) to (w2) {};
		\draw[EdgeStyle] (v1) to (w3) {};
		\draw[EdgeStyle] (v1) to (wnp) {};
		\end{tikzpicture}
		
		\caption{\label{fig:graph-nonbipartite-tight} Nonbipartite graph for which the
			unique solution $c:E\to\R$ satisfies $\sum_{e\in
				E}c^2(e)=\Theta(\antigirth(|V_M|-\antigirth+1)^2\|K\|^2)$. We have distinguished
			with thick lines a unique cycle of the graph}
	\end{figure}

	Let us now present an example that shows tightness of the result above. Let a SAMC $G_M=(V_M,E_M)$ be a graph with
	$V=\{v_1,\dots,v_\antigirth,w_1,\dots,w_{n-\antigirth}\}$ with odd $\antigirth$,
	and edge set constructed as follows:
	\begin{enumerate}
		\item $v_1,\dots,v_\antigirth$ forms a cycle in given order,
		\item for each $i=1,\dots,n-\antigirth$ we have $\{v_1,w_i\}\in E$,
	\end{enumerate}
	see Fig.~\ref{fig:graph-nonbipartite-tight} for visualization.
	The graph $G$ is unicyclic, hence the $c$ function is unique. Let us consider
	conditioning function $K:v\mapsto k>0$. Note that function of the form
	\begin{equation}
	\begin{split}
	c(e) = \begin{cases}
	-k, & e=\{v_1,w_i\},\\
	k(n-\antigirth-1)/2, & e = \{v_\antigirth,v_1\},\\ 
	k(n-\antigirth-1)/2, & e = \{v_i,v_{i+1}\} \textrm{ and $i$ odd},\\
	-k(n-\antigirth-1)/2-k, & e = \{v_i,v_{i+1}\} \textrm{ and $i$ even}
	\end{cases}
	\end{split}
	\end{equation}
	is a unique solution of optimization problem Eq.~\eqref{eq:optimization-problem}. Finally we have
	\begin{equation}
	\begin{split}
	\|c\|_2^2 &= (n-\antigirth)k^2 +
	\frac{\antigirth+1}{2}\left(k(n-\antigirth-1)/2\right)^2+
	\frac{\antigirth-1}{2}\left(-k(n-\antigirth-1)/2-k\right)^2 \\
	&\geq \antigirth k^2 (n-\antigirth -1)^2/4 = \Omega(\antigirth \|K\|^2
	(n-\antigirth+1)^2)
	\end{split}
	\end{equation}
	As we can see in case of $\antigirth \approx n/2$ our theorem is tight. 
	
	\subsection{Success probability for bipartite SAMC}\label{sec:bipartite}
	Finally, we will prove similar bounds for bipartite graphs. The proofs again
	follow the sketch presented on Fig.~\ref{fig:sketch-of-the-proof}. However, 
	bipartite graphs do not have cycles of odd length. This in turn implies
	that for fixed unicyclic graph solution $c$ is no longer unique.
	
	\begin{proposition}\label{theorem:bound-generalk-bipartite}
		Let $G=(V,E)$ be a bipartite SAMC of connected graph with an appropriate conditioning function $K$ and cycle of length $\antigirth$. Then there is a solution $c$
		of optimization problem Eq.~\eqref{eq:optimization-problem} such that
		\begin{equation}
		\|c\|_2^2 \leq \antigirth |V|^2\|K\|^2 + (|V|-\antigirth) (|V|-\antigirth+1)^2
		\|K\|^2
		\end{equation}
		and
		\begin{equation}
		\|c\|_2^2 \leq |V|\|K\|_1^2.
		\end{equation}
	\end{proposition}
	\begin{proof}
		The proofs goes similarly to proof of
		Theorem~\ref{theorem:bound-generalk-nonbipartite}, up to system of linear
		equations. In this case we have an even number of equations with either no
		solutions, or infinite number of solutions. Since $K$ satisfies
		Eq.~\eqref{eq:condition-bipartite-k} it turns that the solution is
		parametrized by single free variable and takes the form
		\begin{equation}
		c(e_j) = (-1)^{j+1}\sum_{k=j}^{\antigirth-1} (-1)^{k}z_k + (-1)^{j+1}z.
		\end{equation}
		where $z$ is a free parameter. Let $z\coloneqq  z_\antigirth$. Then 
		\begin{equation}
		|c(e_j)| \leq \sum_{k=j}^{\antigirth-1} |z_k| + |z_\antigirth|\leq
		\sum_{k=1}^\antigirth |z_k|\leq\sum_{k=1}^\antigirth \|K\|^{V_k}_1 = \|K\|_1.
		\end{equation}
		
		Finally, similarly to the proof of
		Theorem~\ref{theorem:bound-generalk-bipartite}
		\begin{equation}
		\begin{split}
		\|c\|_2^2 &= \sum_{j=1}^\antigirth |c(e_j)|^2 + \sum_{j=1}^\antigirth
		\sum_{e\in E_j} |c(e)|^2 \\
		&\leq \antigirth \|K\|_1^2 + (|V|-\antigirth) \|K\|_1^2 = |V|\|K\|_1^2
		\end{split}
		\end{equation} 
		and
		\begin{equation}
		\begin{split}
		\|c\|_2^2 \leq \antigirth |V|^2\|K\|^2 + (|V|-\antigirth)(|V|-\antigirth+1)^2
		\|K\|^2.
		\end{split}
		\end{equation} 
	\end{proof}
	
	Below we present an upperbound for general bipartite graphs. The proof goes the
	same as it was the non-bipartite case.
	\begin{theorem}
		Let $G_M=(V_M,E_M)$ be a bipartite SAMC of connected graph $G=(V,E)$ with an appropriate conditioning function $K$.
		Then the probability $p_M(t)$ of finding any element from $V_M$ at time $t$
		satisfies 
		\begin{equation}
		\max_{t\geq 0} p_M(t) =\order{\frac{|V_M|^3\|K\|^2}{2|E|-2|E_M| - D^{\bar M}}}.
		\end{equation} 
	\end{theorem}

	Note that in the proof we searched for unicyclic spanning subgraph, while we
	could search for spanning tree and use Theorem~\ref{theorem:tree-useful-bound}.
	Careful analysis would show the same bound $\order{|V_M|^3 \|K\|^2}$, since diameter
	of any spanning tree is at most $|V_M|$.
	
	\subsubsection{Special case: constant conditioning function $K$}
	We can lower the upperbound for bipartite
	SAMC if we assume $K$ to be constant. The proof goes the same way as previous
	proofs, however it differ in details regarding the estimation of the $\|c\|_2^2$. 
	
	\begin{proposition} \label{theorem:bound-constantk-bipartite}
		Let $G=(V,E)$ be a bipartite SAMC of some graph with an appropriate constant
conditioning function $K$ and with cycle of length $\antigirth$. Then there
exists a solution $c$ of optimization problem
Eq.~\eqref{eq:optimization-problem} such that
		\begin{equation}
		\|c\|_2^2 \leq |V|(|V|-\antigirth+1)^2\|K\|^2.
		\end{equation}
	\end{proposition}
	\begin{proof}
		Whole construction of $c$ and other variables remains the same, as in
		Theorem~\ref{theorem:bound-generalk-bipartite}, however we can provide more
		explicit form of $c$ on edges from cycle. Note that according to the
		construction we have
		\begin{equation}
		z_k = -z_j^0 - K',
		\end{equation}
		where $z_j^0 \coloneqq \sum_{e\in I_C^j(v_j)} c(e)$. Thus we can do following
		\begin{equation}
		\begin{split}
		c(e_j) &= (-1)^{j+1}\sum_{k=j}^{\antigirth-1} (-1)^{k}z_k + (-1)^{j+1}z \\
		&= (-1)^{j}\sum_{k=j}^{\antigirth-1} (-1)^{k}(z_k^0+K') + (-1)^{j+1}z\\
		&= (-1)^{j}\sum_{k=j}^{\antigirth-1}
		(-1)^{k}z_k^0+\frac{K'}{2}\left(1-(-1)^{j}\right) + (-1)^{j+1}z
		\end{split}
		\end{equation} 
		Hence for $|z|=|z_1|$ we have
		\begin{equation}
		\begin{split}
		|c(e_i)|\leq\|K'\| + \sum_{k=1}^{\antigirth}(|V_k|-1)\|K\|
		=(|V|-\antigirth+1)\|K\|&.
		\end{split}
		\end{equation}
		Hence, similarly as in Theorem~\ref{theorem:bound-generalk-bipartite}, we have
		\begin{equation}
		\begin{split}
		\|c\|_2^2 &\leq \sum_{j=1}^\antigirth |c(e_j)|^2 + \sum_{j=1}^\antigirth
		\sum_{e\in E_j} |c(e)|^2  \\
		&\leq \antigirth(|V|-\antigirth+1)^2\|K\|^2 + (|V|-\antigirth) (|V|-\antigirth+1)^2
		\|K\|^2 \\
		& =|V|(|V|-\antigirth+1)^2\|K\|^2.
		\end{split}
		\end{equation} 
	\end{proof}
	
	Again the proof goes similar as in previous cases, apart from the fact that upperbound $|E_M| = \order{|V_M|^2}$ is no longer useful here.
	\begin{theorem}
		Let $G_M=(V_M,E_M)$ be a bipartite SAMC of a connected graph $G=(V,E)$ with appropriate constant conditioning function $K$ and cycle of length $\antigirth$.
		Then the probability $p_M(t)$ of finding any element from $V_M$ at time $t$
		satisfies 
		\begin{equation}
		\max_{t\geq 0} p_M(t) =\order{\frac{|V_M|(|V_M|-\antigirth+1)^2\|K\|^2
				+|E_M|}{2|E|-2|E_M| - D^{\bar M}}}.
		\end{equation} 
	\end{theorem}

	\begin{figure}[t]
		\centering
		\begin{tikzpicture}
		\node[VertexStyle] (w1) at (0,2) {$w_1$};
		\node[VertexStyle] (w2) at (0,1) {$w_2$};
		\node[VertexStyle] (w3) at (0,0) {$w_3$};
		\node[] at (0,-.75) {$\vdots$};
		\node[VertexStyle] (wnp) at (0,-1.75) {$w_{n'}$};
		
		\node[VertexStyle] (t1) at (10,2) {$t_1$};
		\node[VertexStyle] (t2) at (10,1) {$t_2$};
		\node[VertexStyle] (t3) at (10,0) {$t_3$};
		\node[] at (10,-.75) {$\vdots$};
		\node[VertexStyle] (tnp) at (10,-1.75) {$t_{n'}$};
		
		\node[VertexStyle] (v1) at (1.5, 0) {$v_1$};
		\node[VertexStyle] (v2) at (2.5,-1) {$v_2$};
		\node[VertexStyle] (v3) at (4,-1) {$v_3$};
		\node[] (vdotsdown) at (5.5, -1) {$\cdots$};
		\node[VertexStyle] (vq2) at (7,-1) {$v_{\frac{\antigirth}{2}}$};
		\node[VertexStyle] (vq2p1) at (8.5,0) {$v_{\frac{\antigirth}{2}+1}$};
		\node[VertexStyle] (vq2p2) at (7, 1) {$v_{\frac{\antigirth}{2}+2}$};
		\node[] (vdotsup) at (5.5, 1) {$\cdots$};
		\node[VertexStyle] (vqm1) at (4,1) {$v_{\antigirth-1}$};
		\node[VertexStyle] (vq) at (2.5,1) {$v_{\antigirth}$};
		
		\draw[EdgeStyle, ultra thick] (v1) to (v2) {};
		\draw[EdgeStyle, ultra thick] (v2) to (v3) {};
		\draw[EdgeStyle, ultra thick] (v3) to (vdotsdown) {} ;
		\draw[EdgeStyle, ultra thick] (vdotsdown) to (vq2) {};
		\draw[EdgeStyle, ultra thick] (vq2) to (vq2p1) {};
		\draw[EdgeStyle, ultra thick] (vq2p1) to (vq2p2) {};
		\draw[EdgeStyle, ultra thick] (vq2p2) to (vdotsup) {};
		\draw[EdgeStyle, ultra thick] (vdotsup) to (vqm1) {};
		\draw[EdgeStyle, ultra thick] (vqm1) to (vq) {};
		\draw[EdgeStyle, ultra thick] (vq) to (v1) {};
		
		\draw[EdgeStyle] (v1) to (w1) {};
		\draw[EdgeStyle] (v1) to (w2) {};
		\draw[EdgeStyle] (v1) to (w3) {};
		\draw[EdgeStyle] (v1) to (wnp) {};
		
		\draw[EdgeStyle] (vq2p1) to (t1) {};
		\draw[EdgeStyle] (vq2p1) to (t2) {};
		\draw[EdgeStyle] (vq2p1) to (t3) {};
		\draw[EdgeStyle] (vq2p1) to (tnp) {};
		\end{tikzpicture}
		
		\caption{\label{fig:graph-bipartite-tight} Bipartite graph for which the
			optimal solution $c:E\to\R$ satisfies $\sum_{e\in
				E}c^2(e)=\Theta(\antigirth(n-\antigirth+1)^2\|K\|^2)$. We distinguish a unique
			cycle with thick lines}
	\end{figure}
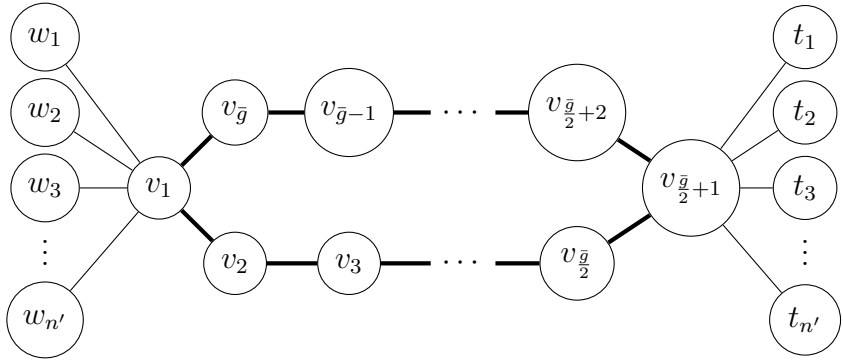
	Now let us present an example which shows our results concerning bipartite
	graphs are almost tight in worst case scenario. Let us consider SAMC with
	vertices
	\begin{equation}
	\begin{split}
	V = \{v_1,\dots,v_{\antigirth},	w_{1},\dots,w_{n'},
	t_{1},\dots,t_{n'},\}
	\end{split}
	\end{equation}
	with $\antigirth$ being even, $\antigirth/2$ odd, with
	$n'=\frac{n-\antigirth}{2}$ and edge set constructed as follows:
	\begin{enumerate}
		\item $v_1,\dots,v_{\antigirth}$ form a cycle in the given order,
		\item for each $i=1,\dots,n'$ we have $\{v_1,w_i\}\in E$ and
		$\{v_{\frac{\antigirth}{2}+1},t_i\}\in E$.
	\end{enumerate}
	Let us consider conditioning function $K:e\mapsto k$ for fixed $k\geq0$. Then $K$ satisfies
	Eq.~\eqref{eq:condition-bipartite-k} and, by the proof of
	Theorem~\ref{theorem:bound-generalk-bipartite}, we have one-variable
	parameterized solution. Let us consider the family of functions constructed as
	in the proof of Theorem~\ref{theorem:bound-generalk-bipartite}
	\begin{equation}
	\begin{split}
	c_z(e) = \begin{cases}
	-k, & e=\{v_{1},w_i\},\\
	-k, & e=\{v_{\frac{\antigirth}{2}+1},t_i\},\\
	k(n'+1)/2+z-k, & e = \{v_i,v_{i+1}\},\ i=1,3,\dots,\antigirth/2,\\
	-k(n'+1)/2-z, & e = \{v_i,v_{i+1}\},\ i=2,4,\dots,\antigirth/2-1,\\
	k(n'+1)/2-z-k, & e = \{v_i,v_{i+1}\},\
	i=\antigirth/2+1,\antigirth/2+3,\dots,\antigirth,\\
	-k(n'+1)/2+z, & e = \{v_i,v_{i+1}\},\
	i=\antigirth/2+2,\antigirth/2+4,\dots,\antigirth-1.\\
	\end{cases}
	\end{split}
	\end{equation}
	With these functions we have
	\begin{equation}
	\begin{split}
	\min_{z\in \R}\| c_z\|_2^2 &=\min_{z\in \R}\Big ( (n-\antigirth)k^2\\
	&\phantom{\ =\min_{z\in \R}\Big (}+ \frac{\antigirth+2}{4}(k(n'+1)/2+z-k)^2\\
	&\phantom{\ =\min_{z\in \R}\Big (}+\frac{\antigirth-2}{4}(-k(n'+1)/2-z)^2\\
	&\phantom{\ =\min_{z\in \R}\Big (}+ \frac{\antigirth+2}{4}(k(n'+1)/2-z-k)^2\\
	&\phantom{\ =\min_{z\in \R}\Big (}+\frac{\antigirth-2}{4}(-k(n'+1)/2+z)^2 \Big
	),
	\end{split}
	\end{equation}
	which is the quadratic function in $z$. One can note, that it takes the form
	$\antigirth z^2  + C$, where $C$ does not depend on $z$. Thus, it takes its
	minimum at $z = 0$ and therefore
	\begin{equation}
	\begin{split}
	\min_{z\in \R} \|c_z\|_2^2 &= (n-\antigirth)k^2 +
	\frac{\antigirth+2}{32}k^2(n-\antigirth-2)^2+\frac{\antigirth-2}{32}k^2(n-\antigirth+2)^2\\
	&= \Theta(\antigirth\|K\|^2 (n-\antigirth+1)^2)
	\end{split}
	\end{equation}

	\subsection{Overview of the results for general
		graphs}\label{sec:overview-general-graphs}
	
	In previous section we present several theorems which, depending on the
	properties of the marked component or $K$ function, yields different
	upperbounds. All of them can be gathered into single suitable equation.
	\begin{theorem}\label{theorem:final-general-k}
		Let $G_M=(V_M,E_M)$ be a SAMC of a connected graph $G=(V,E)$ and an appropriate conditioning function $K$. Let $D^{\bar M}$ denotes the number of edges from $V_M$ to $V \setminus V_M$. Then the probability $p_M(t)$ of finding any element from $V_M$ at time
		$t$ satisfies 
		\begin{equation}
		\max_{t\geq 0} p_M(t) =\order{\frac{|V_M|^3\|K\|^2}{2|E|-2|E_M| - D^{\bar M}}}.
		\end{equation} 
	\end{theorem}
	We would like to emphasize several interesting additional remarks. The proofs of
	the theorem were always constructive. Furthermore, it turned out that for trees
	and unicyclic nonbipartite graphs the function $c$ is unique. These properties
	are particularly helpful, as the optimal $c$ is needed to show the tightness of
	the bound derived in the theorem above. 
		
	Furthermore for general $K$ we have provided similar bounds on $\|c\|_2^2$ in
	terms of $\|K\|_1$  instead of $\|K\|$, which may be more precise when not all
	values are close to the maximum of $K$. 
	
	\section{Special graphs}
	In this section we will provide how we can apply
	Theorem~\ref{theorem:final-general-k} for several important classes of graph. From this moment we will assume that we're given a connected graph $G=(V,E)$ with corresponding SAMC $G_M=(V_M,E_M)$.
	
	First, suppose graph $G$ is bounded-degree, and we have finite order SAMC
$G_M$. Then $V_M$, $E_M$, $D^{\bar M}$ and $\|K\|$ are also bounded by
constant. Furthermore, $|E| = \Theta(n)$. Hence the success probability of
finding any vertex  is $\order{1/n}$.
	
	Similarly, one can show that success probability of finding any vertex from
constant-order SAMC in $d$-regular graphs is $\order{d/n}$. Note that $d$ may
grow with $n$, and unless $d \ll n$, we have that the success probability grows
very slowly. This is expected, as for complete graph, which for our model is
equivalent to Grover search, there is success probability $1-o(1)$ at proper
measurement time.
	
	Finally, let us focus on random graph models. In \cite{glos2019impact} authors considered exceptional configuration as a way to maliciously decrease the efficiency of quantum search. They considered \ER \cite{erdHos1960evolution}, \WS \cite{watts1998collective} and \BA graphs \cite{albert2002statistical} as potential input graphs, on which the quantum search is run. In particular, authors considered what is the efficiency of their attack, for these models, see Fig.~2 in the paper, which was defined as the ratio of the success probability between single marked component and SAMC. While for \ER and \WS models it turned out that the success probability is much smaller for SAMC, for \BA the numerical results were not obvious. In fact, it turned out, that by an appropriate change of measurement time we can keep high efficiency of the search. Below we will present an analytical quantification of this effect.
	
	 Let us start with \ER model $G(n,p)$ of graphs of order $n$, where $p$ is the probability that any two vertices are connected. Note that for $p > (1+\varepsilon)\log(n)/n$ the graphs are connected and almost regular, which means that almost surely for any vertex $v$ we have $\deg(v) = np(1+o(1))$. This means that we can use the same reasoning as it was for regular graphs, and show, that the probability for finite $|V_M|$ is of order $\order{d/n} = \order{np/n} = p$. For sparse graphs, where $np = \order{\log^k(n)}$, the success probability decreases almost like $1/n$. 
	 
	 \WS graphs are transitioned model between \ER graphs and regular ring lattice \cite{watts1998collective}. While contrary to \ER model it produces small-world graphs, the graphs are again almost-regular. This enables us following the above reasoning.
	 
	 Finally let us consider \BA model $G(n,m_0)$. Graphs in this model are constructed by consecutively adding new nodes, and connecting them to $m_0$ already existing nodes. It is the first model which is scale-free, thus not regular or even almost-regular. The minimum degree is always $m_0$, which is constant, while the largest degree $\Delta$ grows like $\Theta(\sqrt{n})$  \cite{flaxman2005high}.
	 
	 Let us consider finite order SAMC . Any graph sampled from \BA distributions has $|E| = nm_0(1+o(1))$ edges. Since $V_M$ is finite, so is $E_M$. Note that $D^{\bar M} \leq |V_M| \Delta = \order{\sqrt{n}}$. This simplifies our bound to
	 \begin{equation}
	 \max_{t\geq 0} p_M(t) =\order{\|K\|^2/n}.
	 \end{equation} 
	 Note that $\|K\|$ equals to the maximal degree over marked vertices. If the marked vertex has a degree $\Theta(\sqrt{n})$, then we obtain $\max_{t\geq 0} p_M(t) =\order{1}$, which does not provide us anything interesting. However, most of the vertices in \BA model are bounded \cite{bollobas2001degree}, and for them the success probability is $\order{1/n}$. This provides us an example of important class, for which the success probability strongly depends on the vertices being marked. We claim that such variety of the results  explains the robustness obtained and presented on Fig.~2 in \cite{glos2019impact}.
	
	\section{Conclusions}
	
	In this paper we presented a general upperbound on success probability, which takes the following form
	\begin{theorem}
		Let $G_M=(V_M,E_M)$ be a SAMC of a connected graph $G=(V,E)$ and an appropriate conditioning function $K$. Then the probability $p_M(t)$ of finding any element from $V_M$ at time
		$t$ satisfies 
		\begin{equation}
		\max_{t\geq 0} p_M(t) =\order{\frac{|V_M|^3\|K\|^2}{2|E|-2|E_M| - D^{\bar M}}}.
		\end{equation} 
	\end{theorem}
	The result can be applied to any graph and marked subgraph combination, as long as the appropriate quantum search operation have a stationary state which is uniform over edges not being marked. We also provided examples, which show our results are tight in complexity.
	
	Finally, we applied our results for a variety of wide class graphs, which includes bounded degree graphs, regular graphs, and selected important random graph models. For many cases we showed that the success probability decreases at least as $\order{1/{|V|}}$. However for \BA model, we have shown that the success probability bound depends on the marked vertices. By this we explained the robustness obtained in \cite{glos2019impact}.
	
	We believe that further analysis of exceptional configurations may improve the already obtained results concerning the malicious input data modification \cite{glos2019impact}, or may help generalizing the algorithms for detecting perfect matching~\cite{nahimovs2019adjacent}.

	\subsubsection*{Acknowledgements} 
	NN is supported by the QuantERA ERA-NET Cofund in Quantum Technologies
	implemented within the European Union's Horizon 2020 Programme (QuantAlgo
	project). KK is uspported by RFBR according to the research project No.
	19-37-80008. AG has been partially supported by National Science
	Center under grant agreement 2019/32/T/ST6/00158. AG would also
	like to acknowledge START scholarship from the Foundation for Polish
	Science.

	
	\bibliographystyle{ieeetr}	
	\bibliography{bound_paper}

\end{document}